%% file: main.tex
\documentclass{article}
\pdfoutput=1
\usepackage{palatino}
\usepackage{euler}
\usepackage{latexsym}
\usepackage{amsmath,amssymb,amstext,amsthm}
\usepackage{graphicx}
\usepackage{tabularx}
\usepackage{stmaryrd}
\usepackage{mathdots}
\usepackage[colorlinks]{hyperref}
\usepackage[german,english]{babel}

\usepackage{tikz}
\usetikzlibrary{arrows,automata,backgrounds,calc,chains,fadings,folding,decorations.fractals,fit,patterns,positioning,mindmap,shadows,shapes.geometric,shapes.symbols,through,trees,decorations.markings,decorations.text}

\input{main.def}

\newcommand{\perioddoubling}{\msf{T}}

\begin{document}

\title{Let's Make a Difference\,!\thanks{%
  An earlier version of this paper appeared in~\cite{la:roel:2009}.%
}}

\author{%
  J\"{o}rg Endrullis \and Dimitri Hendriks \and Jan Willem Klop
}

\date{}

\maketitle

\begin{dedication}
  Dedicated to Roel de Vrijer on the occasion of his 60th birthday.
\end{dedication}

\begin{abstract}
  \input{abstract}
\end{abstract}

\section{Introduction}
\input{intro}

\section{Generalized Difference Operators}
\input{diff}

\section{Periodic Orbits}
\input{periodic-orbits}

\section{The $\sdiff$-orbits of some non-periodic streams: \\ Fibonacci, Mephisto, \Sierpinski}
\input{non-periodic-orbits}

\section{Concluding Remarks}
\input{conclusion}

\bibliography{main}

\end{document}

%% file: abstract.tex
We study the behaviour of iterations of the difference operator~$\sdiff$ 
on streams over $\{0,1\}$.
In particular, we show that a stream $\astr$ is eventually periodic 
if and only if the sequence of differences 
$\astr,\diff{\astr},\diffn{2}{\astr},\ldots$,
the `$\sdiff$-orbit' of $\astr$ as we call it, is eventually periodic.
Moreover, we generalise this result to operations $\sblockdiff{d}$
that sum modulo $2$ the elements of each consecutive block of length $d+1$ 
in a given \mbox{$01$-stream}.
Some experimentation with $\sdiff$-orbits of well-known streams reveals a surprising
connexion between the \Sierpinski{} stream and the Mephisto Waltz.

%% file: intro.tex
In previous work~\cite{{endr:grab:hend:isih:klop:2007},{endr:grab:hend:isih:klop:2009},{endr:grab:hend:2008}} 
we have been interested in definability of $01$-streams by means of fixed point
equations in a certain restricted format (PSF, pure stream format), restricted enough to guarantee
decidability of productivity, a notion of well-definedness. The format PSF was expressive enough
to encompass all automatic sequences~\cite{allo:shal:2003}.
In the course of those investigations we often employed as illustrations some well-known streams, 
such as the \thuemorse{} sequence~$\morse$, 
the Toeplitz or period doubling sequence~$\perioddoubling$, 
the Fibonacci stream~$\fib$, 
the \Sierpinski{} stream~$\sierpinski$, and the Mephisto Waltz~$\mephisto$.
For definitions of these streams see Table~\ref{table:specs};
for more background see~\cite{klop:2009}.	

Apart from the expressivity or definability aspect, we also were and are very interested
in relations between such streams: can we transform one stream into another, 
employing a certain arsenal of transformations --- 
such as e.g.\ finite state transducers (FSTs), or, equivalently, 
unary contexts in the PSF-format.

One striking, well-known transformation is that of $\morse$ into $\perioddoubling$ 
using the `first difference operator' $\sdiff$, 
defined by $\nth{\diff{\astr}}{n} = \nth{\astr}{n} + \nth{\astr}{n+1}$, 
for all $01$-streams~$\astr$ and $n\in\nat$,
where $+$ is addition modulo $2$,
or, in the PSF format:
\begin{align*}
  \diff{\cons{x}{\cons{y}{\astr}}}
  & \to \cons{(x+y)}{\diff{\cons{y}{\astr}}}
\end{align*}
and with an equivalent FST as in Figure~\ref{fig:diff-fst}.
\begin{figure}[ht!]
  \begin{center}
    \scalebox{1}{\includegraphics{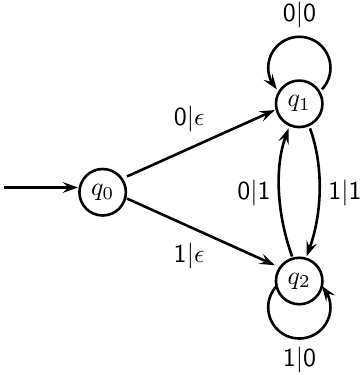}}
  \end{center}
  \caption{%
    \textit{An FST implementing the $\sdiff$ operator.}
  }
  \label{fig:diff-fst}
\end{figure}

Now $\diff{\morse} = \perioddoubling$, as one easily verifies:
\begin{align*}
  \morse    & =   0110  1001  1001  0110  1001  0110  0110  1001 \ldots \\
  \perioddoubling & = \, 101 1 101 0 101 1 101 1 101 1 101 0 101 1 101 \ldots
\end{align*}
A first question now presents itself: what do we encounter
by iterating $\sdiff$, so that we get the `$\sdiff$-orbit' of $\morse$:
\begin{align*}
  \morse, 
  \diff{\morse}, \diffn{2}{\morse}, \diffn{3}{\morse}, \ldots 
\end{align*}
A visual impression is given by Figure~\ref{fig:diffn:morse}.
\begin{figure}[ht!]
  \begin{center}
  \begin{minipage}{.50\textwidth}
  \begin{center}
   \zoom{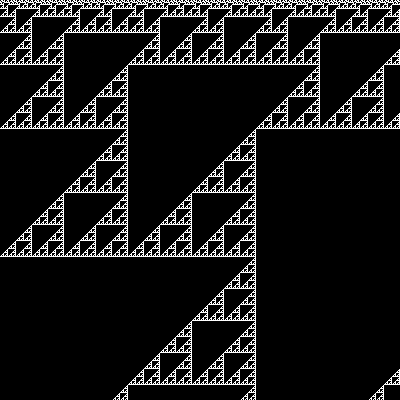}{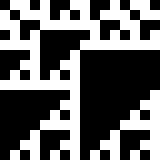}
  \end{center}
  \end{minipage}
  \end{center}
  \caption{%
    \textit{The first 400 iterations of $\sdiff$ on the \thuemorse{} sequence (top row);
      $0$s are black, $1$s are white.}
  }
  \label{fig:diffn:morse}
\end{figure}
We will prove that this $\sdiff$-orbit of $\morse$ is not periodic ---
that is, all streams $\diffn{n}{\morse}$ are mutually different.%
\footnote{%
  In fact, we can give explicit expressions for the iterations of $\sdiff$
  (see~\cite{klop:2002}): \\
  $\diffn{2n}{\morse} = \zip{\diffn{n}{\morse}}{\diffn{n}{\morse}}$
  and $\diffn{2n+1}{\morse} = \zip{\zeros}{\diffn{n+1}{\morse}}$.
}

This non-periodicity fact is a corollary of Theorem~\ref{thm:blockdifforbit:periodicity} 
below stating that the $\sdiff$-orbit $\astr, \diff{\astr}, \diffn{2}{\astr}, \ldots$ 
of an arbitrary stream $\astr\in\str{\{0,1\}}$
is eventually periodic if and only if $\astr$ is eventually periodic. 
We also generalise this periodicity theorem to 
operations~$\sblockdiff{d}$
that we call `$d{+}1$-block difference'; $\sdiff$ is then the $2$-block difference~$\sblockdiff{1}$.

Next, we observe that the difference matrix with top row $\morse$ (see Figure~\ref{fig:diffn:morse}) 
exhibits ever growing triangles of zeros. It is as if repeated application of $\sdiff$ tends 
to damp out the volatility of the stream $\morse$, so that in $\diffn{n}{\morse}$ 
ever larger stretches of $0$'s appear.
We wondered whether this is a general phenomenon, 
and therefore we determined the $\sdiff$-orbit of some other streams,
starting with the Fibonacci stream~$\fib$, see Figure~\ref{fig:diffn:fib}. 
The result is strikingly different from the $\sdiff$-orbit of $\morse$: 
the black triangles now seem uniformly bounded in size. 
So the `damping out' effect that $\sdiff$ had on $\morse$, 
is by no means general. 
We also give an example showing how significant information can be detected 
from a consideration of these `fingerprint' patterns exhibited by the $\sdiff$-orbits, 
displayed as matrices as in Figures~\ref{fig:diffn:morse}, \ref{fig:diffn:kai}, \ref{fig:diffn:fib},
and~\ref{fig:diffn:sierpinski:mephisto}.
Namely, in an experiment it turned out 
(see Figure~\ref{fig:diffn:sierpinski:mephisto})
that the $\sdiff$-matrix of the \Sierpinski{} stream~$\sierpinski$ 
and the Mephisto Waltz~$\mephisto$ of Keane~\cite{kean:1968}
are after the first couple of rows exactly the same! 
In this we way, by comparing these fingerprints, 
we found that
\begin{align*}
	\diffn{2}{\sierpinski} & = \diffn{3}{\mephisto}
\end{align*}
a curious fact that seems hard to find or guess otherwise, 
because $\sierpinski$ and $\mephisto$ seem totally unrelated in their definition.

%% file: diff.tex
The `$d{+}1$-block difference'~$\blockdiff{d}{\astr}$ of a bitstream 
is the stream obtained by adding modulo $2$, 
each block of $d+1$ consecutive elements of $\astr$,
that is:
\begin{align*}
  \nth{\blockdiff{d}{\astr}}{i} 
  = \nth{\astr}{i} + \cdots + \nth{\astr}{i+d}
\end{align*}
So we have a `sliding window' of length $d+1$ moving through the stream~$\astr$.

Some preliminary remarks are in order.
Let $\bit = \{0,1\}$. 
For $a,b\in\bit$, we write $a \xor b$ for the sum (or difference) of $a$ and $b$ modulo~$2$,
and $\inv{a}$ for the inverse of $a$ defined by $\inv{a} = a \xor 1$.
We use $\str{\bit}$ to denote the set 
$\str{\bit} = \{ \astr \where \astr \funin \nat\to\bit \}$ of
infinite words (streams) over the alphabet $\bit$.

\begin{definition}\label{def:diff}\normalfont
  For $d\in\nat$, the 
  operator $\sblockdiff{d}\funin\setop{\str{\bit}}$ is defined as follows:
  \begin{align*}
    \nth{\blockdiff{d}{\astr}}{i} & = \summ{j=0}{d}{\nth{\astr}{i+j}}
    &\text{for all $\astr\in\str{\bit}$ and $i\in\nat$.}
  \end{align*}
  We call $\blockdiff{d}{\astr}$ the \emph{$d+1$-block difference of $\astr$},
  and we define $\sdiff$ to be $\sdiff = \sblockdiff{1}$.
  
  The \emph{$\sblockdiff{d}$-orbit} of a stream~$\astr\in\str{\bit}$,
  which we denote by $\blockdifforbit{d}{\astr}$,
  is defined as the infinite sequence of iterated block differences of $\astr$:
  \begin{align*}
    \blockdifforbit{d}{\astr} 
    & \defdby (\blockdiffn{d}{n}{\astr})_{n=0}^{\infinity}
  \end{align*}
  We write $\difforbit{\astr}$ for $\blockdifforbit{1}{\astr}$.
\end{definition}
There is a close correspondence between 
the $n$-th iteration of $\sdiff$ and the triangle of Pascal:
\begin{align}
  \nth{\diffn{n}{\astr}}{i} = \sum_{k=0}^n \binom{n}{k} \cdot \nth{\astr}{i+k}
  \label{eq:diffn:triangle:correspondence}
\end{align}
A quick hint for this is obtained by inspecting the $\sdiff$-orbit of a stream $\pointstream{p}$
which is $0$ everywhere except at position $p$; 
so let $\pointstream{p}\in\str{\bit}$ be defined by 
\begin{align*}
  \nth{\pointstream{p}}{p} = 1 && \text{and} && \nth{\pointstream{p}}{n} = 0 \quad \text{if $n\neq p$.}
\end{align*}
Figure~\ref{fig:diffn:kai} 
\begin{figure}[ht!]
  \begin{center}
    \scalebox{.4}{\includegraphics{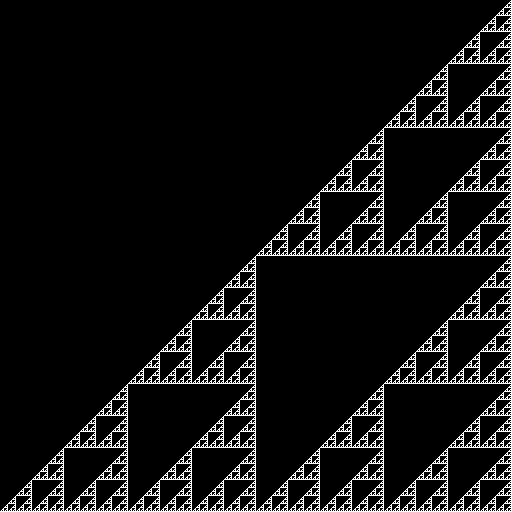}}
  \end{center}
  \caption{%
    \textit{A $512{\times}512$-cut of the $\sdiff$-orbit $\difforbit{\hspace{.08em}\pointstream{511}}$.
    }
  }
  \label{fig:diffn:kai}
\end{figure}
pictures the $p{\times}p$-cut of the $\sdiff$-orbit $\difforbit{\pointstream{p-1}}$, for $p = 2^9$. 
Not displayed is the infinite darkness to the right. 
Vertically the figure repeats itself, as we explain below.
We clearly find back, though somewhat slanted, the \Sierpinski{} triangle,
which is the limit, as the number of rows approaches infinity, of Pascal's triangle modulo $2$.
The last row $\diffn{p-1}{\pointstream{p-1}}$ in Figure~\ref{fig:diffn:kai} 
consists of $p$ ones (followed by infinitely many zeros), 
and so the first row not in the picture will be equal to $\pointstream{p-1}$ (top-row) again.
Hence, the period of this $\sdiff$-orbit is $p$.

To obtain a more general result for $\sblockdiff{d}$
similar to~\eqref{eq:diffn:triangle:correspondence},
we use a generalisation of Pascal's triangle.
First we recall the recursion equation for $\binom{n}{k}$, 
the entry at row~$n$, column~$k$ of Pascal's triangle:
\begin{align*}
  \binom{n}{k} & = \binom{n-1}{k} + \binom{n-1}{k-1}
\end{align*}
adding up the values in columns $k$ and $k-1$ of the previous row $n-1$.
In triangles $\spascal{d}$ defined below, with $d\in\nat$, 
the value $\pascal{d}{n}{k}$ at row~$n$, column~$k$ is the sum 
$\pascal{d}{n}{k} = \pascal{d}{n-1}{k} + \pascal{d}{n-1}{k-1} + \cdots + \pascal{d}{n-1}{k-d}$.
\begin{definition}\label{def:pascal}\normalfont
  Let $d\in\nat$. We define \emph{triangle $\spascal{d}\funin\nat\times\ints\to\nat$}
  as follows:
  \begin{align*}
    \pascal{d}{0}{0} & = 1 \\
    \pascal{d}{n}{k} 
    & = \summ{i=0}{d}{\pascal{d}{n-1}{k-i}} & (0 \leq k \leq dn , n > 0) \\
    \pascal{d}{n}{k} & = 0 & (k < 0 \,\text{ or }\, k \gt dn)
  \end{align*}
  and we write $\pascalrow{d}{n}$ for the \emph{$n$-th row} of triangle $\spascal{d}$, 
  that is, for the sequence:
  \begin{align*}
    \pascalrow{d}{n} & = \pascal{d}{n}{0} \,,\, \ldots \,,\, \pascal{d}{n}{dn}
  \end{align*}
\end{definition}
Note that triangle~$\spascal{1}$ is the usual Pascal triangle:
\[
  \pascal{1}{n}{k} = \binom{n}{k}
\]

As an example, the first couple of rows of triangle~$\spascal{2}$
are shown in Figure~\ref{fig:pascal2}.
\begin{figure}[h!]
  \newcommand{\A}{\iddots}
  \newcommand{\B}{\vdots}
  \newcommand{\C}{\ddots}
  \newcolumntype{Y}{>{$}c<{$}} 
  \centering{
  \scalebox{.8}{
  \begin{tabularx}{1\linewidth}{YYYYYYYYYYYYYYY}
       &   &   &    &    &    &     &   1 &     &    &    &    &   &   & \\ 
       &   &   &    &    &    &   1 &   1 &   1 &    &    &    &   &   & \\ 
       &   &   &    &    &  1 &   2 &   3 &   2 &  1 &    &    &   &   & \\ 
       &   &   &    &  1 &  3 &   6 &   7 &   6 &  3 &  1 &    &   &   & \\ 
       &   &   &  1 &  4 & 10 &  16 &  19 &  16 & 10 &  4 &  1 &   &   & \\ 
       &   & 1 &  5 & 15 & 30 &  45 &  51 &  45 & 30 & 15 &  5 & 1 &   & \\
       & 1 & 6 & 21 & 50 & 90 & 126 & 141 & 126 & 90 & 50 & 21 & 6 & 1 & \\[-.5ex]      
    \A &   &   &    &    &    &     &  \B &     &    &    &    &   &   & \C     
  \end{tabularx}
  }
  }
  \caption{\textit{Triangle $\spascal{2}$.}} 
  \label{fig:pascal2}
\end{figure}
\begin{remark}\label{rem:pascal:clean}
  Perhaps a cleaner, but notationally heavier, definition of $\pascal{d}{n}{k}$ 
  is to define it only for values $0\leq k\leq dn$,
  as follows:
  \begin{align}
    \pascal{d}{n}{k}
    & = \summ{i=\max(0,k-d(n-1))}{\min(d,k)}{\pascal{d}{n-1}{k-i}}
    \tag{$\sharp$}
    \label{eq:pascal:clean}
  \end{align}
  so that the summation is only over the non-zero (defined) values of the previous row,
  as we have $0 \leq k-i \leq d(n-1)$ due to the range of the index variable~$i$.
  The sum in~\eqref{eq:pascal:clean} is easily seen to be equal to the sum
  $\summ{i=0}{d}{\pascal{d}{n-1}{k-i}}$ of Definition~\ref{def:pascal}:
  in the latter we allow the 
  variable $i$ to also go through 
  the values smaller than $k-d(n-1)$, and values greater than $k$. 
  This does not change the outcome, 
  because for these values of $i$ we get that $k-i \gt d(n-1)$, 
  and $k-i \lt 0$, respectively, and thus $\pascal{d}{n-1}{k-i} = 0$ 
  by definition of $\spascal{d}$.
\end{remark}

In the area of combinatorial mathematics 
Pascal's triangle has been generalised in the way we do here, 
e.g.\ in~\cite{freund:1956}, where $\pascal{d}{n}{k}$ means (our notation):
\begin{quotation}
  \noindent
  \textit{%
  $\pascal{d}{n}{k}$ is the number of distinct ways in which $k$ indistinguishable objects 
  can be distributed in $n$ cells allowing at most $d$ objects per cell.
  }~\cite{freund:1956}
\end{quotation}

\newcommand{\agraph}{G}
\newcommand{\iagraph}{\sub{\agraph}}
\newcommand{\avertices}{V}
\newcommand{\iavertices}{\sub{\avertices}}
\newcommand{\aedges}{E}
\newcommand{\iaedges}{\sub{\aedges}}
Alternatively, one can view a triangle $\spascal{d}$ as a graph
$\iagraph{d} = \pair{\iavertices{d}}{\iaedges{d}}$ where:
\begin{align*}
  \iavertices{d} & = \{\, \pair{n}{k} \where n\in\nat ,\, 0 \leq k \leq dn \,\} \\
  \iaedges{d}    & = \{\, \pair{\pair{n}{k}}{\pair{n+1}{k+j}} \where 0 \leq j \leq d \,\}
\end{align*}
Then the value $\pascal{d}{n}{k}$ is the number of paths from $\pair{0}{0}$, 
the root of $\iagraph{d}$, to the vertex~$\pair{n}{k}$.

The following fact about double summations 
will be encountered several times in the sequel.
For all $\chi\funin\nat\times\nat\to\nat$, 
and $p,q\in\nat$:
\begin{align}
  \summ{i=0}{p}{\summ{j=0}{q}{\,\funap{\chi}{i,j}}}
  = \summ{\phantom{(}k=0\phantom{)}}{p+q\vphantom{()}}{\summ{i=\max(0,k-q)}{\min(p,k)}{\funap{\chi}{i,k-i}}}
  \tag{$\Sigma\!\Sigma$}
  \label{eq:swap:sums}
\end{align}
See for instance the proof of Lemma~\ref{lem:blockdiffn:triangle:correspondence},
where \eqref{eq:swap:sums} is used to move a subexpression of $\chi$ 
which only contains index variable $j$, out of the scope of the summation that binds variable $i$.

The following lemma generalises a familiar property of Pascal's triangle,
namely that the sum of values in the $n$-th row equals $2^n$.
%
%
\begin{lemma}
  \[
    \summ{k=0}{dn}{\pascal{d}{n}{k}} = d^n
  \]
\end{lemma}
\begin{proof}
  By induction on $n$. 
  The case $n=0$ is trivial, as we have that $\pascal{d}{0}{0} = 1$ by definition.
  If $n=\prim{n}+1$ we reason as follows:
  \begin{align*}
    \summ{k=0}{dn}{\pascal{d}{n}{k}}
    & = \summ{k=0\vphantom{()}}{d\prim{n}+d\vphantom{()}}{\summ{i=\max(0,k-d\prim{n})}{\min(d,k)}{\pascal{d}{\prim{n}}{k-i}}} & \text{\eqref{eq:pascal:clean}} \\ 
    & = \summ{i=0}{d}{\summ{j=0}{d\prim{n}}{\pascal{d}{\prim{n}}{j}}} & \text{\eqref{eq:swap:sums}} \\
    & = d \cdot \summ{j=0}{d\prim{n}}{\pascal{d}{\prim{n}}{j}} \\
    & = d \cdot d^{\prim{n}}  &\text{(IH)}
  \end{align*}
\end{proof}

We are ready for the generalisation of~\eqref{eq:diffn:triangle:correspondence} 
on page~\pageref{eq:diffn:triangle:correspondence}.
The triangle $\spascal{d}$ (modulo $2$) can be used 
to relate the $n$-th block difference~$\blockdiffn{d}{n}{\astr}$ 
to the original stream $\astr$, as follows:
\begin{lemma}\label{lem:blockdiffn:triangle:correspondence}
  \begin{align*}
    \nth{\blockdiffn{d}{n}{\astr}}{i}
    & = \summ{k=0}{dn}{\pascal{d}{n}{k}} \cdot \nth{\astr}{i+k}
  \end{align*}
\end{lemma}
\begin{proof}
  By induction on $n$.
  For $n=0$ the statement directly follows by unfolding definitions.
  In case $n = \prim{n}+1$, we reason as follows:
  \begin{align*}
    \nth{\blockdiffn{d}{n}{\astr}}{i}
    & = \nth{\blockdiff{d}{\blockdiffn{d}{\prim{n}}{\astr}}}{i} \\ \displaybreak[0]
    & = \summ{j=0}{d}{\nth{\blockdiffn{d}{\prim{n}}{\astr}}{i+j}} \\ \displaybreak[0]
    & = \summ{j=0}{d}{\summ{k=0}{d\prim{n}}{\pascal{d}{\prim{n}}{k}} \cdot \nth{\astr}{i+j+k}} & \text{(IH)}\\ \displaybreak[0]
    & = \summ{\ell=0}{dn}{\summ{j=\max(0,\ell-d\prim{n})}{\min(d,\ell)}{\pascal{d}{\prim{n}}{\ell-j}\cdot\nth{\astr}{i+\ell}}} & \text{\eqref{eq:swap:sums}} \\ \displaybreak[0]
    & = \summ{\ell=0}{dn}{\summ{j=0}{d}{\pascal{d}{\prim{n}}{\ell-j}\cdot\nth{\astr}{i+\ell}}}  & \text{\eqref{eq:pascal:clean}} \\ \displaybreak[0] 
    & = \summ{\ell=0}{dn}{\Big(\summ{j=0}{d}{\pascal{d}{\prim{n}}{\ell-j}\Big)\cdot\nth{\astr}{i+\ell}}} \\ \displaybreak[0]
    & = \summ{\ell=0}{dn}{\pascal{d}{n}{\ell}\cdot\nth{\astr}{i+\ell}}
  \end{align*}
\end{proof}

%% file: periodic-orbits.tex
In this section we show that the $\sdiff$-orbit $\blockdifforbit{d}{\astr}$ of a bitstream $\astr\in\str{\bit}$ 
is eventually periodic if and only if the stream $\astr$ itself is eventually periodic.

As an instance of Lemma~\ref{lem:blockdiffn:triangle:correspondence} 
we obtain that by taking $\spascal{2}$ in Figure~\ref{fig:pascal2} 
we derive how the values of $\sblockdiffn{2}{4}$ are related 
to those of the original $\astr$, as follows:
\begin{align*}
  \blockdiffn{2}{4}{i} 
  &
  =
  1 \cdot \nth{\astr}{i}
    +  4 \cdot \nth{\astr}{i+1}
    + 10 \cdot \nth{\astr}{i+2}
    + 16 \cdot \nth{\astr}{i+3}
    + 19 \cdot \nth{\astr}{i+4} \\
  &
    \mathrel{\phantom{=}}
    \mbox{}
    + 16 \cdot \nth{\astr}{i+5}
    + 10 \cdot \nth{\astr}{i+6}
    +  4 \cdot \nth{\astr}{i+7}
    +  1 \cdot \nth{\astr}{i+8} \\
  &
  =
  \nth{\astr}{i} + \nth{\astr}{i+4} + \nth{\astr}{i+8}
\end{align*}
as only for $k = 0,4,8$ the entries $\pascal{2}{4}{k}$
are odd.
For iterations which are powers of $2$
this can be generalised, because, for $n=2^m$, 
the $n$-th row of a triangle $\spascal{d}$ modulo $2$
always has the following shape:
\begin{align*}
  \pascalrow{d}{n} = 1 \, 0^{n-1} \, 1 \, 0^{n-1} \, \cdots \, 0^{n-1} \, 1
  &&
  (n = 2^m)
\end{align*}
that is, $d+1$ many $1$s with blocks $0^{n-1}$ in between them.
This observation, translated to $\sdiff$-orbits in the following lemma, 
is crucial for the main result of this paper, Theorem~\ref{thm:blockdifforbit:periodicity},
for it enables us to pinpoint the periodicity in the orbit by looking at rows $2^m$
with $\congrmod{2^m}{0}{p}$, 
with $p$ the period of $\astr$.
\begin{lemma}\label{lem:blockdiffn{d}{2^m}}
  For $n$\, a power of $2$, we have:
  \begin{align*}
    \nth{\blockdiffn{d}{n}{\astr}}{i}
    & =
    \summ{j=0}{d}{\nth{\astr}{i+jn}}
  \end{align*}
\end{lemma}
\begin{proof}
  Let $n=2^m$. The proof proceeds by induction on $m$.
  The base case $m=0$ follows directly by definition of $\sblockdiff{d}$.
  If $m = \prim{m}+1$, we let $\prim{n} = 2^{\prim{m}}$ and we infer:
  \begin{align*}
    \nth{\blockdiffn{d}{n}{\astr}}{i}
    & = \nth{\blockdiffn{d}{\prim{n}}{\blockdiffn{d}{\prim{n}}{\astr}}}{i} \\
    & = \summ{j=0}{d}{\nth{\blockdiffn{d}{\prim{n}}{\astr}}{i + j\cdot \prim{n}}} & \text{(IH)} \\
    & = \summ{j=0}{d}{\summ{k=0}{d}{\nth{\astr}{i + (j + k) \cdot \prim{n}}}} & \text{(IH)} \\
    & = \summ{\phantom{(}\ell=0\phantom{)}}{2d\vphantom{(}}{\summ{j=\max(0,\ell-d)}{\min(d,\ell)}{\nth{\astr}{i + \ell\cdot\prim{n}}}} & \text{\eqref{eq:swap:sums}} \\
    & = \summ{\ell=0}{2d}{\Big(\summ{j=\max(0,\ell-d)}{\min(d,\ell)}{1}\Big)\cdot{\nth{\astr}{i + \ell\cdot\prim{n}}}}
  \end{align*}  
  Let us abbreviate the subexpression $\summ{j=\max(0,\ell-d)}{\min(d,\ell)}{1}$
  by $\funap{S}{\ell}$.
  For both $\ell \leq d$ and $\ell \gt d$ we have that 
  $\funap{S}{\ell} = \ell + 1$, and hence we can continue as follows:
  \begin{align*}
    & = \summ{\ell=0}{2d}{\funap{S}{\ell}\cdot{\nth{\astr}{i + \ell\cdot\prim{n}}}} \\
    & = \summ{\ell=0}{d}{\funap{S}{2\ell}\cdot{\nth{\astr}{i + 2\ell\cdot\prim{n}}}}
        + \summ{\ell=0}{d-1}{\funap{S}{2\ell+1}\cdot{\nth{\astr}{i + (2\ell+1)\cdot\prim{n}}}} \\
    & = \summ{\ell=0}{d}{1\cdot{\nth{\astr}{i + 2\ell\cdot\prim{n}}}}
        + \summ{\ell=0}{d-1}{0\cdot{\nth{\astr}{i + (2\ell+1)\cdot\prim{n}}}} \\
    & = \summ{\ell=0}{d}{\nth{\astr}{i + \ell\cdot n}}
  \end{align*}
  Thus, we have shown, for $n=2^m$, the equality
  $\nth{\blockdiffn{d}{n}{\astr}}{i} = \summ{\ell=0}{d}{\nth{\astr}{i + \ell\cdot n}}$.
\end{proof}

For $d=1$, Lemma~\ref{lem:blockdiffn{d}{2^m}} gives 
$\nth{\diffn{n}{\astr}}{i} = \nth{\astr}{i} + \nth{\astr}{i+n}$, for $n=2^m$.
Indeed, for $n=2^m$ all $\binom{n}{k}$ for $1\leq k\leq n-1$ are even, 
and this in turn follows from the fact that 
all rows $\pascalrow{1}{n}$ with $n=2^{m}-1$ 
consist of odd numbers only. 
From the latter observation we also obtain:
\begin{lemma}\label{lem:diffn:blockdiff}
  If $n=2^{m}-1$ for some $m\in\nat$, 
  then $\sdiffn{n} = \sblockdiff{n}$.
\end{lemma}
\begin{proof}
  For all $m,n\in\nat$, we prove that if $n=2^m-1$ then 
  $\nth{\diffn{n}{\astr}}{i} = \nth{\blockdiff{n}{\astr}}{i}$,
  for all $\astr\in\str{\bit}$ and $i\in\nat$,
  by induction on $m$.
  If $m=0$, then $n=0$ and we have $\diffn{0}{\astr} = \astr = \blockdiff{0}{\astr}$.
  If $m=\prim{m}+1$, then $n = 2\prim{n}+1$ with $\prim{n} = 2^{\prim{m}}-1$,
  and we reason as follows:
  \begin{align*}
    \nth{\diffn{n}{\astr}}{i}
    & = \nth{\diff{\diffn{\prim{n}}{\diffn{\prim{n}}{\astr}}}}{i} \\
    & = \nth{\diffn{\prim{n}}{\diffn{\prim{n}}{\astr}}}{i} 
        + \nth{\diffn{\prim{n}}{\diffn{\prim{n}}{\astr}}}{i+1} \\ 
    & = \nth{\blockdiff{\prim{n}}{\blockdiff{\prim{n}}{\astr}}}{i} 
        + \nth{\blockdiff{\prim{n}}{\blockdiff{\prim{n}}{\astr}}}{i+1}
        & \text{($4\times\text{IH}$)} \\
    & = \nth{\blockdiffn{\prim{n}}{2}{\astr}}{i} 
        + \nth{\blockdiffn{\prim{n}}{2}{\astr}}{i+1} \\
    & = \summ{j=0}{\prim{n}}{\nth{\astr}{i+2j}}
        + \summ{j=0}{\prim{n}}{\nth{\astr}{i+2j+1}}	
        & \text{($2\times\text{Lemma~\ref{lem:blockdiffn{d}{2^m}}}$)} \\
    & = \summ{j=0}{2\prim{n}+1}{\nth{\astr}{i+j}} \\
    & = \nth{\blockdiff{n}{\astr}}{i}
  \end{align*}

\end{proof}

\begin{definition}
  Let $\aset\neq\emptyset$.
  A sequence $\astr\in\str{\aset}$ is \emph{(eventually) periodic} 
  if there exist $p\geq 1$ and $n_0\in\nat$ such that 
  $\myall{n\geq n_0}{\nth{\astr}{n+p} = \nth{\astr}{n}}$,
  where we call $p$ the \emph{period}, and $n_0$ the \emph{offset} of $\astr$.
  %
  For a function $f\funin\str{\aset}\to\str{\aset}$, 
  we say that \emph{$f$ strongly preserves} 
  periodicity
  if $\funap{f}{\astr}$ has the same period and offset as $\astr$.
\end{definition}
In other words, $\astr$ is eventually periodic 
if and only if there exist $p\geq 1$ and $n_0\in\nat$ 
such that for all $n_1, n_2 \geq n_0$ with $\congrmod{n_1}{n_2}{p}$
we have $\nth{\astr}{n_1} = \nth{\astr}{n_2}$.
Note that we do not require $p$ and $n_0$ to be minimal.

Let us drop the adjective `eventual'
and take `periodic' to mean `eventually periodic'.

\begin{lemma}\label{lem:blockdiff:preserves:period}
  The difference operator $\sblockdiff{d}$ strongly preserves periodicity.
\end{lemma}
\begin{proof}
  Let $\astr\in\str{\bit}$ be a periodic stream with period $p\in\nat$
  and offset~$n_0\in\nat$.
  Then it immediately follows that:
  \begin{align*}
    \nth{\blockdiff{d}{\astr}}{n+p} 
    = \summ{j=0}{d}{\nth{\astr}{n+p+j}}
    = \summ{j=0}{d}{\nth{\astr}{n+j}}
    = \nth{\blockdiff{d}{\astr}}{n} 
  \end{align*}
  for all $n \geq n_0$, 
  and hence $\blockdiff{d}{\astr}$ is periodic with period~$p$ and offset~$n_0$.
\end{proof}
\noindent
By linearity of $\sblockdiff{d}$ we obtain that all $\sblockdiffn{d}{m}$ 
strongly preserve periodicity.

We come to the main theorem of our contribution,
roughly stating that horizontal periodicity of a $\sdiff$-orbit 
implies vertical periodicity, and vice versa.

\begin{theorem}\label{thm:blockdifforbit:periodicity}
  Let $d\in\nat$. 
  A stream $\astr\in\str{\bit}$ is periodic 
  if and only if 
  $\blockdifforbit{d}{\astr}$ is periodic.
\end{theorem}

\proof
  Let $d\in\nat$. We prove both implications separately.
  \begin{itemize}

  \item[\textit{(only if)}]

	Let $\astr$ be a periodic stream
	with period $p$ and offset~$n_0$.
	Furthermore, let $N_1 = 2^{m_1}$ and $N_2 = 2^{m_2}$
	for some $m_1,m_2\in\nat$ such that
	$n_0 < N_1 < N_2$ and $\congrmod{N_1}{N_2}{p}$
	(these are bound to exist, as there are finitely many equivalence classes 
	$\{ m \where \congrmod{m}{n}{p} \}$
	and infinitely many powers of~$2$).
	Then, by Lemma~\ref{lem:blockdiffn{d}{2^m}} 
	and periodicity of $\astr$, we find:
	\begin{align*}
	  \nth{\blockdiffn{d}{N_1}{\astr}}{n} 
	    = \summ{j=0}{d}{\nth{\astr}{n+jN_1}}
	  & = \summ{j=0}{d}{\nth{\astr}{n+jN_2}}
	    = \nth{\blockdiffn{d}{N_2}{\astr}}{n} 
	\end{align*}
	for all $n \geq n_0$,
	and hence 
	$\blockdifforbit{d}{\astr}$ is periodic.

  \item[\textit{(if)}]

	Let $\blockdifforbit{d}{\astr}$ 
	be periodic with period $p$ and offset~$n_0$,
	i.e., $\blockdiffn{d}{n}{\astr} = \blockdiffn{d}{n+p}{\astr}$, for all $n\geq n_0$.
	Again, let $N_1 = 2^{m_1}$ and $N_2 = 2^{m_2}$
	for some $m_1,m_2\in\nat$ such that
	$n_0 < N_1 < N_2$ and $\congrmod{N_1}{N_2}{p}$.
	Then we have, for all $i\in\nat$, $\nth{\blockdiffn{d}{N_1}{\astr}}{i} = \nth{\blockdiffn{d}{N_2}{\astr}}{i}$.
	From Lemma~\ref{lem:blockdiffn{d}{2^m}} it then follows that
	$\summ{j=0}{d}{\nth{\astr}{i + jN_1}} = \summ{j=0}{d}{\nth{\astr}{i + jN_2}}$.
	Hence, we obtain, still for all $i\in\nat$:
	\begin{align*}
	  \nth{\astr}{i + dN_2} & = \summ{j=0}{d}{\nth{\astr}{i + jN_1}} + \summ{j=0}{d-1}{\nth{\astr}{i + jN_2}}
	\end{align*}
	In other words, the element $\nth{\astr}{i + dN_2}$ 
	is uniquely determined by the $dN_2$ preceding elements of $\astr$.
	We conclude by observing that there are only finitely many blocks of length $dN_2$,
	and hence there must come a repetition: 
	i.e.\ we find $dN_2 \leq i_1 \lt i_2$ 
	such that $\nth{\astr}{i_1+n} = \nth{\astr}{i_2+n}$, for all $n\in\nat$.
	Hence $\astr$ is periodic.
	\qed
  \end{itemize}

%% file: non-periodic-orbits.tex
We define some non-periodic streams and look at their $\sdiff$-orbits.
In Table~\ref{table:specs} 
\begin{table}[ht!]
  \input{specs}
  \caption{\textit{Stream specifications.}}
  \label{table:specs}
\end{table}
we give PSF specifications of 
the \thuemorse{} sequence~$\morse$, the period doubling sequence~$\perioddoubling$, 
the Fibonacci word~$\fib$, the Mephisto Waltz~$\mephisto$, 
and the stream $S$ which we call the \emph{\Sierpinski{} stream}.
(Of course, alternative specifications exist.)

As far as we know, the \Sierpinski{} stream does not occur in the literature.
We have derived it from the construction of the `\Sierpinski{} arrowhead curve', 
see Figure~\ref{fig:sierpinsky}.
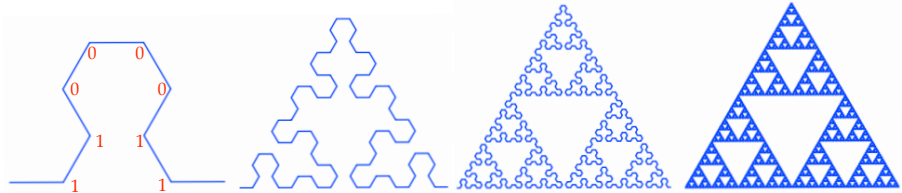
\begin{figure}[h!]
\rotatebox{0.25}{\input{figs/sierpinski-curve}}
\caption{\textit{Construction of the \Sierpinski{} arrowhead curve.}}
\label{fig:sierpinsky}
\end{figure}
The curve is obtained back from the stream~$\sierpinski$
by interpreting its entries 
\( 
  \sierpinski = 
  1 1 0 0 0 0 1 1 1 \,
  0 0 1 1 1 1 0 0 1 \,
  1 1 0 0 0 0 1 1 0 \,
  0 0 1 1 1 1 0 0 0 \,
  1 1 0 0 0 0 1 1 0 \,
  \ldots
\) 
as turtle drawing instructions: 
$1$ means move forward one unit length and turn to the left $\pi/3$,
and $0$ means move forward one unit length and turn to the right $\pi/3$.
In this way, the \Sierpinski{} curve arises as the Hausdorff limit of the finite approximations
(scaling back in size when necessary).

Theorem~\ref{thm:blockdifforbit:periodicity} implies that any stream $\astr$ 
which is equal to one of its differences $\blockdiffn{d}{n}{\astr}$
is periodic. 
Put differently, no two differences of a non-periodic stream, 
e.g.\ the \thuemorse{} sequence~$\morse$, 
are the same.
The $\sdiff$-orbit of $\morse$ is depicted in Figure~\ref{fig:diffn:morse}.
Observe that subsequences of consecutive $0$s become larger and larger.

We do not see this `calming down' aspect in Figure~\ref{fig:diffn:fib},
\begin{figure}[ht!]
  \begin{center}
    \scalebox{.5}{\includegraphics{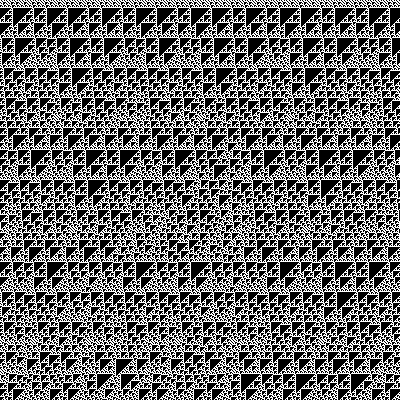}}
  \end{center}
  \caption{%
    \textit{The first $400$ differences of the Fibonacci stream.
    }
  }
  \label{fig:diffn:fib}
\end{figure}
which displays the $\sdiff$-orbit of the Fibonacci stream 
that can be defined as the fixed point of the substitution 
$0\to1,1\to10$ starting on $1$.

One more experiment with $\sdiff$-orbits is shown 
in Figure~\ref{fig:diffn:sierpinski:mephisto}, 
\begin{figure}[ht!]
  \begin{minipage}{.48\textwidth}
  \begin{center}
   \zoom{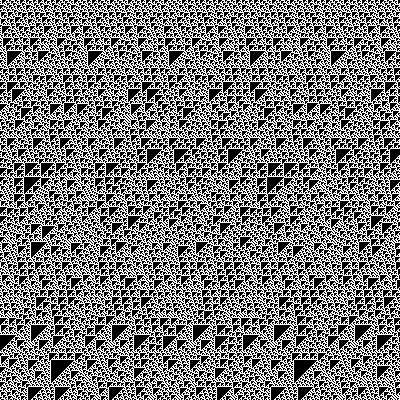}{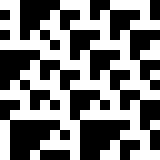}
  \end{center}
  \end{minipage}
  \hfill
  \begin{minipage}{.48\textwidth}
  \begin{center}
   \zoom{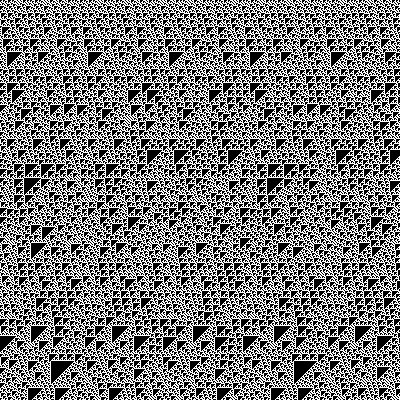}{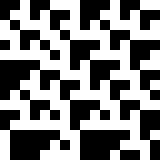}
  \end{center}
  \end{minipage}
  \caption{
    \textit{%
	  Comparing the `fingerprints' $\difforbit{\sierpinski}$ and $\difforbit{\mephisto}$
	  of the \Sierpinski{} stream~$\sierpinski$ (left), and the Mephisto Waltz~$\mephisto$.
	  We find that $\diffn{2}{\sierpinski} = \diffn{3}{\mephisto}$!%
    }
  }
  \label{fig:diffn:sierpinski:mephisto}
\end{figure}
where the $\sdiff$-orbits of the \Sierpinski{} stream~$\sierpinski$ 
and the Mephisto Waltz~$\mephisto$ are displayed.
It is readily seen that both patterns seem identical, 
from the distribution of the black triangles.
That they are indeed identical is revealed by a closer inspection 
of the first couple of rows; it turns out that the third row of the left orbit, 
i.e.\ $\diffn{2}{\sierpinski}$, is identical to the fourth row of the right orbit, 
i.e.\ $\diffn{3}{\mephisto}$.
Indeed, the $16{\times}16$ enlargements show at these row-positions 
both the prefix $1100110111100111$ of length 16.

%% file: specs.tex
\newcommand{\h}[1]{\funap{\sub{\msf{h}}{#1}}}
\newcommand{\hmorse}{\h{\morse}}
\newcommand{\hmephisto}{\h{\mephisto}}
\newcommand{\hfib}{\h{\fib}}
\newcommand{\htoeplitz}{\h{\perioddoubling}}
\newcommand{\sierpinskiseed}{\sub{\msf{w}}{\sierpinski}}
\newcommand{\toepseed}{\sub{\msf{w}}{\perioddoubling}}
\renewcommand{\inv}{\funap{\msf{inv}}}

%
\begin{align*}
  \morse & = \cons{0}{\zipnm{1}{1}{\inv{\morse}}{\tail{\morse}}} 
  \\[1ex]
  \zipnm{n}{m}{\astr}{\bstr} &= \cons{\take{n}{\astr}}{\zipnm{m}{n}{\bstr}{\tailn{n}{\astr}}}
  \\
  \tail{\astr} &= \astr
  \\
  \inv{\cons{0}{\astr}} &= \cons{1}{\inv{\astr}}
  \\
  \inv{\cons{1}{\astr}} &= \cons{0}{\inv{\astr}}
  \\[1ex]
  \perioddoubling & = \zipnm{3}{1}{\toepseed}{\perioddoubling}
  \\
  \toepseed & = \cons{1}{\cons{0}{\cons{1}{\toepseed}}}
  \\[1ex]
  \mephisto & = \hmephisto{\cons{0}{\tail{\mephisto}}} 
  \\
  \hmephisto{\cons{0}{\astr}} &= \cons{0}{\cons{0}{\cons{1}{\hmephisto{\astr}}}} 
  \\
  \hmephisto{\cons{1}{\astr}} &= \cons{1}{\cons{1}{\cons{0}{\hmephisto{\astr}}}} 
  \\[1ex]
  \sierpinski &= \zipnm{8}{1}{\sierpinskiseed}{\sierpinski} 
  \\
  \sierpinskiseed & = \cons{1}{\cons{1}{\cons{0}{\cons{0}{\cons{0}{\cons{0}{\cons{1}{\cons{1}{\inv{\sierpinskiseed}}}}}}}}} 
  \\[1ex]
  \fib & = \hfib{\cons{1}{\tail{\fib}}} 
  \\
  \hfib{\cons{1}{\astr}} &= \cons{1}{\cons{0}{\hfib{\astr}}}
  \\
  \hfib{\cons{0}{\astr}} &= \cons{1}{\hfib{\astr}}
\end{align*}

%% file: figs/sierpinski-curve.tex
\begin{minipage}{.45\textwidth}
\begin{center}
  \scalebox{.35}{
  \includegraphics{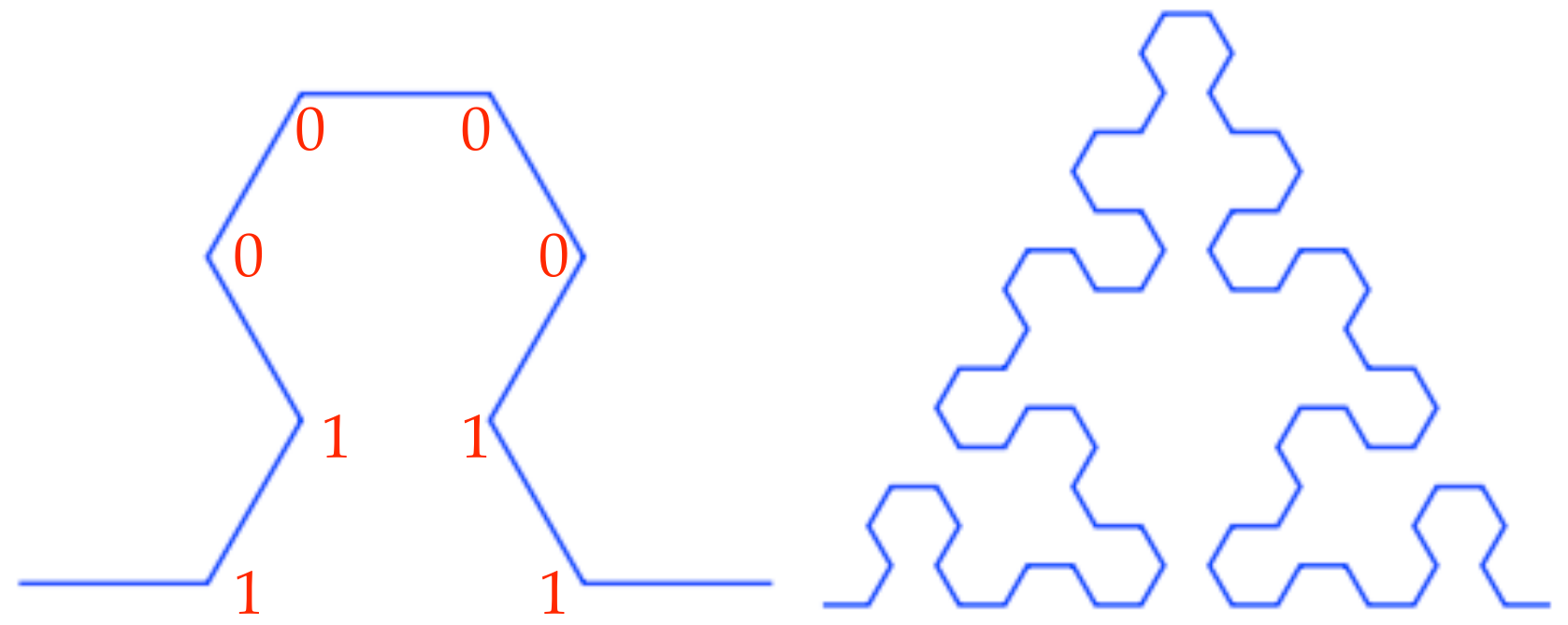}
  }
  \vspace{-4ex}
\end{center}
\end{minipage}
\hspace{1em}
\begin{minipage}{.45\textwidth}
\begin{center}
  \scalebox{.35}{
  \includegraphics{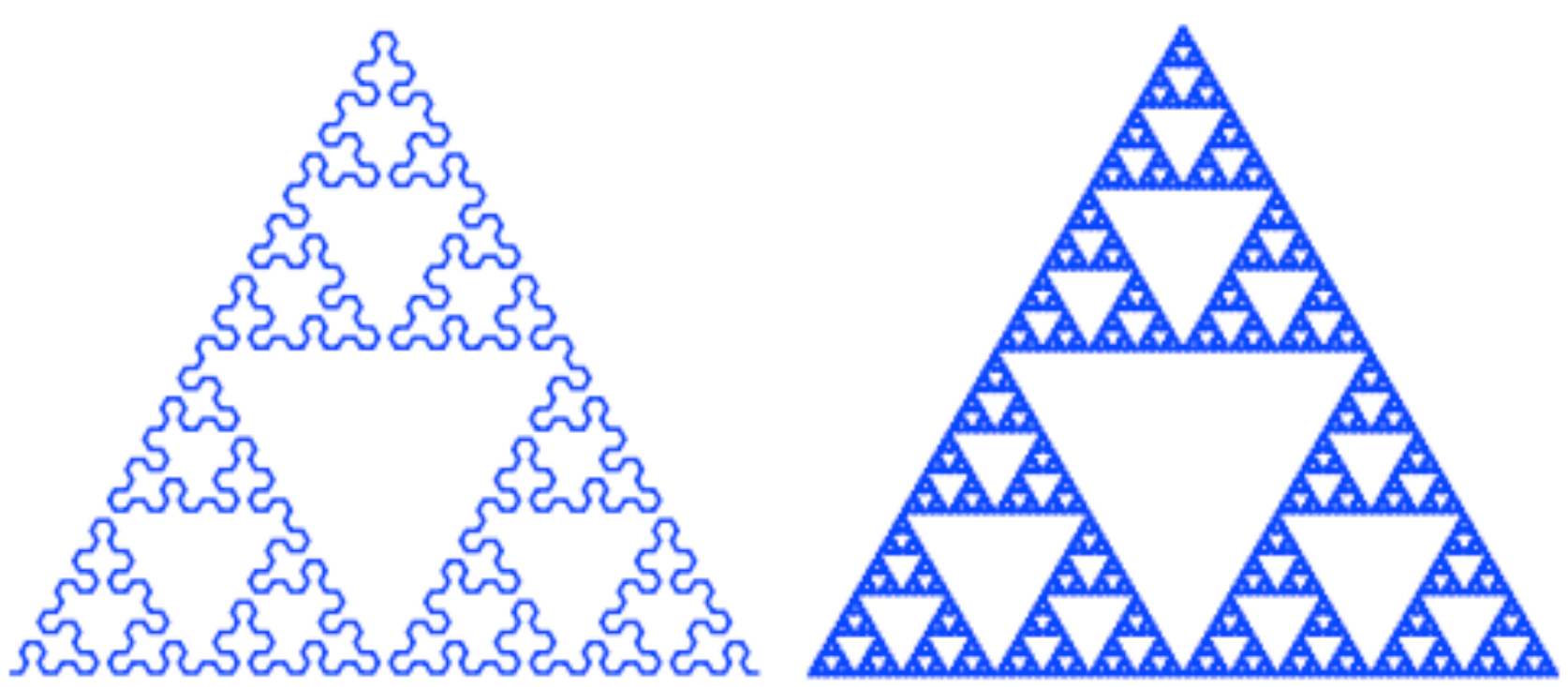}
  }
\end{center}
\end{minipage}

%% file: conclusion.tex
\begin{enumerate}

  \item
	An interesting specific question is whether the dynamical system with as universe
	the $01$-streams and $\sdiff$ as iterator function, is \emph{chaotic} --- 
	such as the dynamical system of $01$-streams with `tail' or `shift' is, as is well-known;
	see e.g.\ \cite[p.~118, Coroll.~11.22]{holm:2000} or~\cite{bank:drag:jone:2003}.
	To this end it suffices to show that the set of points (streams) periodical under $\sdiff$ are dense 
	in the set of all streams, and second that there exists a stream whose $\sdiff$-orbit lies dense 
	in the set of all streams, thus ensuring the topological transitivity of the iterator function $\sdiff$. 
	The third ingredient necessary for $\sdiff$ to be chaotic, namely sensitive dependence on initial 
	conditions, seems clearly to be the case.

  \item
	In general it would be interesting to investigate typical questions in \emph{symbolic dynamics} (see, e.g., \cite{holm:2000})
	for the dynamical systems formed by infinite streams, equip\-ped with continuous stream
	functions that are PSF- or FST-definable.

  \item\label{item:conclusion:3}
    Note that the $\stail$-orbit starting with $\morse$ exhibits the phenomenon of
	\emph{almost periodicity}: in the usual metric on infinite streams, 
	$\tailn{n}{\morse}$ can be made arbitrarily close to $\morse$, 
	by choosing $n$ large enough. The proof is simple. 
   
   Also the orbit $\blockdifforbit{2}{\morse}$
   seems to be almost periodic.

  \item
    The observation in~\eqref{item:conclusion:3} leads to the following question: 
    for which FST-definable operations
    and which starting streams is the orbit
    almost periodic?

\end{enumerate}